\documentclass[12pt]{article}

\usepackage{fullpage}
\usepackage[T1]{fontenc}
\usepackage{ae,aecompl}
\usepackage[dvips]{graphicx}
\usepackage{amssymb}
\usepackage{amsmath}
\usepackage[round,authoryear]{natbib}
\usepackage{color}
\usepackage{array}
\definecolor{webgreen}{rgb}{0,0.4,0}
\definecolor{webbrown}{rgb}{0.6,0,0}
\definecolor{purple}{rgb}{0.5,0,0.25}
\definecolor{darkblue}{rgb}{0,0,0.7}
\definecolor{darkred}{rgb}{0.7,0,0}
\definecolor{darkgreen}{rgb}{0,0.7,0}
\usepackage[pdfborder=false]{hyperref}
\hypersetup{colorlinks,citecolor=darkred,filecolor=black,linkcolor=darkblue,urlcolor=webgreen,pdfpagemode=none,
pdfstartview=FitH}
\newcommand{\ignore}[1]{}
\newtheorem{lemma}{{\sc Lemma}}[section]

\newtheorem{cor}{{\sc Corollary}}[section]
\newtheorem{theorem}{{\sc Theorem}}[section]
\newtheorem{defn}{{\sc Definition}}[section]

\newenvironment{proof}{\noindent {\bf \sl Proof\/}:\enspace}
{\hfill $\blacksquare{}$ \vspace{12pt}}
\usepackage{sectsty}
\sectionfont{\centering\normalfont\scshape}
\subsectionfont{\centering\normalfont}

\begin{document}
	
	\title{\bf Approximate Revenue from Finite Range Mechanisms}
	\author{\bf Mridu Prabal Goswami\thanks{Indian Statistical Institute, Tezpur, India. I am thankful to a research grant from SERB, DST, Government of India, and Magesh Kumar K. K. and Manish Yadav.}}
	\maketitle

\begin{abstract} \noindent We consider an economic environment where a seller wants to sell an indivisible unit of good to a buyer. We show that revenue from any strategy-proof and individually rational mechanism defined on closed intervals of rich single crossing domains considered in \citep{Goswami1}, can be approximated by the revenue from a sequence of 
strategy-proof and individually rational mechanisms with finite range. Thus while studying optimal mechanisms without loss of generality we can study mechanisms with finite range.   	    
\end{abstract}

Key words: optimal mechanism, strategy-proof, finite, approximation
\section{\bf Introduction}

We consider an economic environment where a seller wishes to sell an indivisible object to a buyer. The buyer obtains the object with probability $q$, and in return pays $t$ to the seller. The preferences of the buyer are continuous and monotone and come from a domain of rich single-crossing preferences. Various classes of quasilinear and non-quasilinear preference domains are examples of single-crossing domains. 
We assume that the buyer knows, and only the buyer knows her preference, i.e., preference is private information of the buyer.
The seller does not know the preference of the buyer, and thus wants to elicit this private information.     
Hence, we study mechanisms that are strategy-proof, i.e., mechanisms where the buyer has no incentive to misreport her preference. 
\citep{myer} and \citep{Hag and Rog} for quasilinear preferences find the optimal  mechanisms to be finite. Further, \citep{Mishra1} and \citep{Goswami1} study strategy-proof mechanisms in non-quasilinear preferences when the mechanisms have finite range. However, it remains to be investigated whether it is without loss of generality to concentrate on mechanisms that have finite range.                    
The objective of this paper is to establish that the expected revenue from every strategy-proof and individually rational  mechanism defined on rich single-crossing domains of continuous and monotone preferences can be approximated by the expected revenue from a sequence of strategy-proof and individually rational mechanisms with finite range. 
Thus  we can look for optimal mechanisms without loss of generality within the classes of strategy-proof and individually rational mechanisms with finite range.
For a detailed review of literature for mechanism design on single-crossing domains see \citep{Goswami1}. To the best of our information the fact that the expected revenue from a strategy-proof and individually rational mechanism can be approximated by a sequence of expected revenues entailed by a sequence of strategy-proof and individually rational mechanisms with finite range is not known. Applications of single-crossing conditions in mechanism design goes back to \citep{Spence}, \citep{Mirr}       
and \citep{Roths}. However, the strength of this ordinal property to study strategy-proofness in general had remained to be explored.   
For example, to the best of our information, the fact that the single-crossing property is a unifying factor among various quasilinear, non-quasilinear, multidimensional parametric classes of preferences had not been explored until now. In \citep{Goswami1} we give examples to highlight this fact. The geometric structure of strategy-proof mechanisms  obtained due to the single-crossing property in \citep{Goswami1} is used to provide a simple proof of our approximation result. 
The single-crossing property entails an order on the set of preferences. As a result of strategy-proofness the payment function is monotone, where the order on the preferences is due to the single-crossing property. We consider $t$ to be measurable, where the domain of $t$, i.e., a domain of single-crossing preferences, is endowed with the order topology.               
The proof relies on two simple facts: $(1)$ $t$ can be approximated by an increasing sequence of simple functions where each simple function is monotone increasing, $(2)$ every such simple function, they are in fact step functions since $t$ is monotone, is dominated by a strategy-proof mechanism with finite range. Fact $(2)$ is entailed by the single-crossing property. Section \ref{sec:pre} contains basic notations and definitions. In Section \ref{sec:result} we state and prove the main result.    
Section \ref{sec:con} concludes the paper.   
 
\section{\bf Preliminaries} 
\label{sec:pre}

The economic environment in this paper is same as \citep{Goswami1}. It consists of a seller and a buyer. 
The seller sells an indivisible unit of a good. Let $q\in [0,1]$ denote the probability that the object is sold to the buyer.    
In return, the buyer needs to make a payment to the seller. This  payment is denoted by $t$. The set of allocations is denoted by 
$\mathbb{Z}$, and $\mathbb{Z}=[0,\infty[\times [0,1]$, where $[0,\infty[=\Re_{+}$ denotes the set of non-negative real numbers.
A typical bundle is denoted by 
$(t,q)$, where $t\in \Re_{+}$ and $q\in [0,1]$.  The buyer's preference over $\mathbb{Z}$ is denoted by $R$. 
The strict counterpart of $R$ is denoted by $P$, and indifference is denoted by $I$.
For $z\in \mathbb{Z}$, and $R$, let $UC(R,z)=\{z'\in \mathbb{Z}| z' Rz\}$. In words $UC(R,z)$ is the set of bundles that are weakly preferred to $z$ under $R$.
Likewise $LC(R,z)=\{z'|zRz'\}$, $LC(R,z)$ is the set of bundles that are weakly less preferred to $z$ under $R$. 
We assume the preferences to be continuous and monotone, in short we call the CM preferences.       
The notion of a CM preference is defined formally below.

\begin{defn}[ CM Preference] \rm 
	The complete, transitive preference relation $R$ on $\mathbb{Z}$ is \textbf{CM} if $R$ is  
	{\bf monotone}, i.e., 
	
	\begin{itemize}
		
		\item \textbf{money-monotone:} for all $q\in [0,1]$, if $t''>t'$, then $(t',q)P(t'',q)$.
		\item \textbf{$q$-monotone:} for all $t\in \Re_+$, if $q''>q'$, then $(t,q'')P(t,q')$.

	\end{itemize}
	
	\noindent and  \textbf{continuous} for each $z\in \mathbb{Z}$, the sets $UC(R,z)$ and $LC(R,z)$ are closed sets.{\footnote{These two sets are closed in the product topology on the Euclidean space $\mathbb{Z}$.    }}
	
	\label{defn:CM} 
\end{defn}

\noindent The CM preferences are named classical in \citep{Goswami1}. Let $R$ be a CM preference and $x$ a bundle, define  $IC(R,x)=\{y\in \mathbb{Z}\mid yIx\}$.{\footnote{An $IC(R,x)$ set also represents an equivalence class of the equivalence relation $I$.}} The set $IC(R,x)$ is the set of bundles that are indifferent to $x$ according to the preference $R$. 
It can be seen easily that due to the properties of a CM preference, an $IC$ set can be represented as a curve in $\mathbb{Z}$. Thus, we may also call an $IC$ set an $IC$ curve. We shall represent $\Re_{+}$ on the horizontal axis, and $[0,1]$ on the vertical axis. Next we make a remark about considering CM preferences to be primitives of our approach. 

An $IC$ curve in $\mathbb{Z}$ is an upward slopping curve, i.e., if $(t',q'),(t'',q'')\in IC(R,x)$ and $t'<t''$, then $q'<q''$. Let $x'=(t',q'), x''=(t'',q'')$. By $x'\leq x''$ we mean either $x'=x''$ or $t'<t'',q'<q''$. Further, by $x'<x''$ we mean $t'<t'', q'<q''$.     
We call two bundles $x'=(t',q'), x''=(t'',q'')$ {\bf diagonal} if $x'< x''$. 
The single-crossing property is defined next.

\begin{defn}[Single-Crossing of two Preferences]\rm
	We say that two distinct CM preferences $R', R''$ exhibit the {\bf single-crossing property} if and only if  for all $x,y,z \in \mathbb{Z}$, 
	$$\text{if}~ z\in IC(R',x)\cap IC(R'',y),~\text{then}~ IC(R',x)\cap IC(R'',y)=\{z\}.$$
	\label{defn:single_crossing}
\end{defn}
\noindent The single-crossing property implies that two $IC$ curves of two distinct preferences can meet ( or cut) at most one bundle. The single-crossing property is an ordinal property of preferences, i.e., this property does not depend on utility representations of preferences.
Next we define the notion of a single-crossing domain.

\begin{defn}[Rich Single-crossing domain]\rm We call a subset of the set of CM preferences {\bf single-crossing domain} if any $R',R''$ that belongs to the subset 
	satisfy the single-crossing property. We call a single crossing domain {\bf rich} if for any two bundles $x'=(t',q'), x''=(t'',q'')$ such that $t'<t'',q'<q''$ there is $R$ in the single crossing domain such that 
	$x'Ix''$. We denote a rich single-crossing domain by $\mathcal{R}^{cms}$.        	
\end{defn}

\noindent The single-crossing property provides a natural way to define an order on $\mathcal{R}^{cms}$. We state the order next. Let $\square(z)=\{x\mid x\leq z\}$.

\begin{defn}\rm 
	Let $\mathcal{R}^{cms}$ be a rich single crossing domain. Consider $z\in \mathbb{Z}$ and $R', R''\in \mathcal{R}^{cms}$. We say, $R''$ \textbf{cuts} $R'$ \textbf{from above} at $z\in \mathbb{Z}$, if and only if  
	\[ \square(z)\cap UC(R'',z) \subseteq \square(z)\cap UC(R',z). \]
	\noindent We say that $R''$ cuts $R'$ from above if $R''$ cuts $R'$ from above at every bundle. In this case we say $R'\prec R''$. 
	\label{defn:cut} 
\end{defn}

\noindent In words, the indifference curve of $R''$ through $z$ lies above the indifference curve for $R'$ through $z$ in $\square(z)$ if both indifference curves are viewed from the horizontal axis $\Re_+$. It is argued in \citep{Goswami1}  that if $R''$ cuts $R'$ from above at some $z\in \mathbb{Z}$, then $R''$ cuts $R'$ from above at every $z\in \mathbb{Z}$ which makes the order well defined. 
Due to the order $\prec$, the intervals $[R^1, R^2]=\{R, R^1, R^2\in \mathcal{R}^{cms}\mid R^1\precsim R\precsim R^2\}, [R^1, R^2[=\{R, R^1,R^2\in \mathcal{R}^{cms}\mid R^1\precsim R\prec R^2\}$,  $]R^1, R^2[=\{R, R^1, R^2\in \mathcal{R}^{cms}\mid R^1\prec R\prec R^2\}$, $]R^1, R^2]=\{R, R^1, R^2\in \mathcal{R}^{cms}\mid R^1\prec R\precsim R^2\}$ are well defined.
In this order $\mathcal{R}^{cms}$ is a linear continuum and thus 
supremuma and infimum are well defined. 
The order $\prec$ defines an order topology on $\mathcal{R}^{cms}$.  
This order topology is metrizable, and intervals are connected and closed intervals are compact in $\mathcal{R}^{cms}$, see \citep{Goswami1} for details.
Thus the sigma algebra on any topological subspace of $\mathcal{R}^{cms}$ is the corresponding Borel sigma algebra of the subspace.   
We define  strategy-proof and individually rational mechanisms next.

\begin{defn} \rm A function $F:[\underline{R}, \overline{R}]\rightarrow \mathbb{Z}$ is called a {\bf mechanism}. A mechanism is {\bf strategy-proof} if for every $R', R''\in [\underline{R}, \overline{R}]$,
	$F(R')R'F(R'')$.  A mechanism is {\bf individually rational} if for all $R\in [\underline{R}, \overline{R}]$,  $F(R)R(0,0)$.    
	\label{defn:sp}
\end{defn}

\noindent In terms of component functions we denote a mechanism $F$ by $F=(t_F,q_F)$.

\begin{defn}\rm 
	A SCF $F:\mathcal{R}^{cms}\rightarrow \mathbb{Z}$ is \textbf{monotone with respect to the order relation} $\prec$ on $\mathcal{R}^{cms}$ or simply \textbf{monotone} if for every $R', R''\in \mathcal{R}^{cms}$,
	[$R'\prec R''\iff F(R') \leq F(R'')$].
	\label{defn:mon}
\end{defn}

\noindent The following lemma shows that if $F$ is strategy-proof, then $F$ is monotone. 
\begin{lemma}\rm 
	Let  $F:\mathcal{R}^{cms}\rightarrow \mathbb{Z}$ be a strategy proof SCF. Then $F$ is monotone.
	\label{lemma:mon}
\end{lemma}

\begin{proof}\rm See \citep{Goswami1}

\end{proof}

\section{\bf The Result}
\label{sec:result}
Denote $([\underline{R}, \overline{R}],B([\underline{R}, \overline{R}]), \mu)$ be a probability space, where $B([\underline{R}, \overline{R}])$ is the Borel sigma algebra due to the order topology, and $\mu$ is a probability measure. We assume $\mu $ to be $(i)$ ({\bf non-trivial}) for any non-empty and non-singleton interval $A\subseteq [\underline{R},\overline{R}]$ that is  $\mu(A)>0$ ; and $(ii)$ ({\bf continuous}) for any $A\in {\cal B}$ with $\mu(A)>0$ and any $c\in \Re$ with $0<c<\mu(A)$, there exists $B\in {\cal B}$ with $B\subseteq A$ and $\mu(B)=c$. Continuity of $\mu$ implies $\mu$ is non-atomic, in particular $\mu(\{R\})=0$ for all $R\in [\underline{R},\overline{R}]$. Let $F=(t_F,q_F)$, then $E[F]=\int_{[\underline{R}, \overline{R}]}t_F(R)d\mu(R)$.  
Since $\mu(\underline{R})=0$ and we are interested in mechanisms that maximize expected revenue, without loss of generality let $F(\underline{R})=(0,0)$. If for a strategy-proof and individually rational mechanism $F$, there is a mechanism $G$ with 
$Rn(G)$, where $Rn(G)$ finite denotes the range of $G$, and $E[F]\leq E[G]$, the we have nothing to prove. Thus, in 
Theorem \ref{thm:result} we assume that there is no strategy-proof mechanism $G$ with $Rn(G)$ finite that can generate at least as large expected revenue generated by $F$.

\begin{theorem}\rm Let $F=(t_F,q_F)$ be strategy-proof and individually rational.  
Suppose for all mechanisms $G$   
with $Rn(G)$ finite, $E[G]<E[F]$. Then, for every $\epsilon>0$ there is a mechanism $G$ with $Rn(G)$ finite such that $E[F]-E[G]<\epsilon$. 
\label{thm:result}
\end{theorem}

\begin{proof} If $Rn(F)$ is finite, then we have nothing to prove.    
	
\noindent{\bf Step $1$:} We note that $t_F$ can be approximated by a sequence of simple functions $\{s_{n}\}_{n=1}^{\infty}, s_{n}:[\underline{R},\overline{R}]\rightarrow [0,\infty[$ such that if $R'\prec  R''$, then $s_n(R')\leq s_n(R'')$. Also for all $n$, for all $R$, $s_n(R)\leq t_F(R)$.
This follows from the construction of the simple functions. Let $k=1,\ldots, n2^n$

$$s_n(R)=\begin{cases}
	\frac{k-1}{2^n}, & \text{if $\frac{k-1}{2^n}\leq t_F(R)<\frac{k}{2^n}$;}\\
	n, & \text{if $t_F(R)\geq n$.}
\end{cases}$$

\noindent The construction of these simple functions follows from Theorem $1.5.5$ in \citep{Ash}.   
      
\medskip      
      
\noindent{\bf Step $2$:} Consider the case when $F$ is a continuous, so that $t_F$ is continuous. We investigate the nature of $s_n$s in more details.
For this let without loss of generality $2\leq t_F(\overline{R})$, the number $2$ is chosen arbitrarily to see how the simple functions behave.   
Suppose $n=1$. Then $k=1, 2$. 
Since $t_F$ is monotone, we can consider a measurable partition of $[\underline{R},\overline{R}]$ such that $0\leq t_F(R)<\frac{1}{2}$ if $R\in [\underline{R}, R^1[$; $\frac{1}{2}\leq t_F(R)<1$ if $R\in [R^1, R^2[$, 
and $t_{F}(F)\geq 1$ if $R\in [R^2, \overline{R}]$.       
Then $s_1(R)=0$ for all $R\in [\underline{R}, R^1[$, and $s_1(R)=\frac{1}{2}$ for all $R\in [R^1,R^2[$ and $s_{1}(R)=1$  $R\in [R^2, \overline{R}]$. Since $[\underline{R}, \overline{R}]$ is endowed with the Borel sigma algebra, $[\underline{R}, R^1[$, $[R^1, R^2[$, $[R^2, \overline{R}]$ are Borel measurable. 
Now let $n=2$. Then $k=1, 2, \ldots, 8$. 
Then, the partition of the range of $t_F$ is 
$[0, \frac{1}{4}[, [   \frac{1}{4},  \frac{2}{4}=\frac{1}{2}[, [   \frac{2}{4}, \frac{3}{4}[, [\frac{3}{4}, 1[, [1, \frac{5}{4}[, [\frac{5}{4}, \frac{6}{4}[,  
[\frac{6}{4}, \frac{7}{4}[,  [\frac{7}{4}, 2[, [2,\infty[$. 
Since $t_F$ is monotone we can set 

$$\begin{cases}
	0\leq t_F(R)<\frac {1}{4}, & \text{if $R\in [\underline{R}, R^1[$;}\\
	\frac {1}{4}\leq  t_F(R)<\frac {1}{2}, & \text{if $[R^1, R^2[$;}\\
\frac {1}{2}\leq  t_F(R)<\frac {3}{4}, & \text{if $R\in[R^2, R^3[$}\\
\frac {3}{4}\leq  t_F(R)< 1 & \text{if $R\in[R^3, R^4[  $}\\ 
1\leq  t_F(R)< \frac{5}{4}, & \text{if $R\in[R^4, R^5[  $}\\ 
\frac {5}{4}\leq  t_F(R)< \frac{6}{4}, & \text{if$R\in[R^5, R^6[ $}\\ 
\frac {6}{4}\leq  t_F(R)< \frac{7}{4}, & \text{if $R\in[R^6, R^7[ $}\\ 
\frac {7}{4}\leq  t_F(R)< 2, & \text{if $R\in[R^7, R^8[ $}\\ 
2\leq  t_F(R), & \text{if $R\in [R^8,\overline{R}] $.}

\end{cases}$$

\noindent The above partition is well defined. 
To see this consider without loss of generality $R_{*}=\inf \{R\mid t(R)=\frac{1}{4}\}$. 
Then, since $[\underline{R},\overline{R}]$ is metrizable, consider a sequence $R^n\rightarrow R_*$ and $t(R^n)=\frac{1}{4}$. Then by continuity of $t$, $t(R_*)=\frac{1}{4}$. That is, $R_*=R^1$.    
Thus, 

$$s_{2}(R)=\begin{cases}
	0, & \text{if $R\in [\underline{R}, R^1[$;}\\
	\frac {1}{4}, & \text{if $[R^1, R^2[$;}\\
	\frac {1}{2}, & \text{if $R\in[R^2, R^3[$;}\\
	\frac {3}{4}, & \text{if $R\in[R^3, R^4[ $;}\\ 
	1, & \text{if $R\in[R^4, R^5[  $;}\\ 
	\frac {5}{4}, & \text{if$R\in[R^5, R^6[ $;}\\ 
	\frac {6}{4}, & \text{if $R\in[R^6, R^7[ $;}\\ 
	\frac {7}{4}, & \text{if $R\in[R^7, R^8[ $;}\\ 
	2 & \text{if $R\in [R^8,\overline{R}] $.}
	
\end{cases}$$

\noindent We observe that each $s_n$ is increasing and that in each measurable interval $s_n$ takes the smallest value of $t_F$ in that interval. By Theorem $1.5.5$
in \citep{Ash}, for all $R$, $\lim_{n\rightarrow \infty}s_n(R)=t_{F}(R)$.  
By the monotone convergence theorem in \citep{Ash} $\lim_{n\rightarrow \infty}\int_{[\underline{R},\overline{R}]}s_{n}(R)d\mu(R)=
\int_{[\underline{R},\overline{R}]}t_F(R)d\mu(R)$. Without loss of generality consider $s_2$ and the measurable partition that defines $s_2$, and consider the following mechanism with finite range. 

$$G^{s_2}(R)=\begin{cases}
	(0,0), & \text{if $R\in [\underline{R}, R^{1*}[$;}\\
	F(R^1), & \text{if $[R^{1*}, R^{2*}[$;}\\
	F(R^2), & \text{if $R\in[R^{2*}, R^{3*}[$;}\\
	F(R^3), & \text{if $R\in[R^{3*}, R^{4*}[ $;}\\ 
	F(R^4), & \text{if $R\in[R^{4*}, R^{5*}[  $;}\\ 
	F(R^5), & \text{if$R\in[R^{5*}, R^{6*}[ $;}\\ 
	F(R^6), & \text{if $R\in[R^{6*}, R^{7*}[ $;}\\ 
	F(R^7), & \text{if $R\in[R^{7*}, R^{8*}[ $;}\\ 
	F(R^8), & \text{if $R\in [R^{8*}, R^{9*}]; $}\\
	F(\overline{R}), & \text{if $R\in [R^{9*}, \overline{R}]$}. 
\end{cases}$$

\noindent Here $G^{s_2}=(t_G^{s_2}, q_G^{s_2})$, and $(0,0)I^{1*}F(R^1)I^{2*}F(R^2),\ldots, F(R^8)I^{9*}F(\overline{R})$.
Further, $\underline{R}\precsim R^{1*}\precsim,\ldots,\precsim R^{9*}\precsim \overline{R}$. 
Figure $1$ depicts this construction. Since $F$ is strategy-proof and individually rational due to the single-crossing property it follows in a straightforward manner that $G^{s_2}$ is strategy-proof and individually rational.  Since, $\underline{R}\precsim R^{1*}\precsim R^1\precsim R^{2*}\precsim R^2\precsim R^{3*}\precsim R^4 \precsim R^{4*}\precsim R^5\precsim R^{5*}\precsim R^6 \precsim R^{6*}\precsim R^{7}\precsim R^{7*}\precsim R^{8}\precsim R^{8*}\precsim R^{9}\precsim R^{9*}\precsim \overline{R}$; for all $R$, $s_{2}(R)\leq t_{G}^{s_2}(R)$. Consider $R\in [\underline{R}, R^1]$. For all $R\in [\underline{R}, R^{1*}[$, $s_{2}(R)=0=t_{G}^{s_2}(R)$; 
and for all $R\in [R^{1*}, R^{1}]$, $s_{2}(R) \in \{0,\frac{1}{4}\}\leq t_{G}^{s_2}(R)=t_{F}(R^1)$. 
For all $R \in [R^1,R^{2*}[$, $s_{2}(R)=\frac{1}{4}\leq t_{G}^{s_2}(R)=t_{F}(R^1)$; and  
for all $R \in [R^{2*}, R^{2}]$, $s_{2}(R)\in \{\frac{1}{4}, \frac{1}{2}\}\leq t_{G}^{s_2}(R)=t_{F}(R^2)$ and so on. Thus, 
$\int_{[\underline{R}, \overline{R}]}s_{2}(R)d\mu(R)\leq \int_{[\underline{R}, \overline{R}]}t_{G}^{s_2}(R)d\mu(R)$, and in general for all $n$,  
$\int_{[\underline{R}, \overline{R}]}s_{n}(R)d\mu(R)\leq \int_{[\underline{R}, \overline{R}]}t_{G}^{s_n}(R)d\mu(R)\leq \int_{[\underline{R}, \overline{R}]}t_F(R)d\mu(R)$.  
Since $\lim_{n\rightarrow \infty}\int_{[\underline{R}, \overline{R}]}s_{n}(R)d\mu(R)= \int_{[\underline{R}, \overline{R}]}t_F(R)d\mu(R)$, there exists $G^{s_n}$ such that
$E(F)-E(G^{s_n})<\epsilon$.

\begin{center}
\includegraphics[height=12cm, width=16cm]{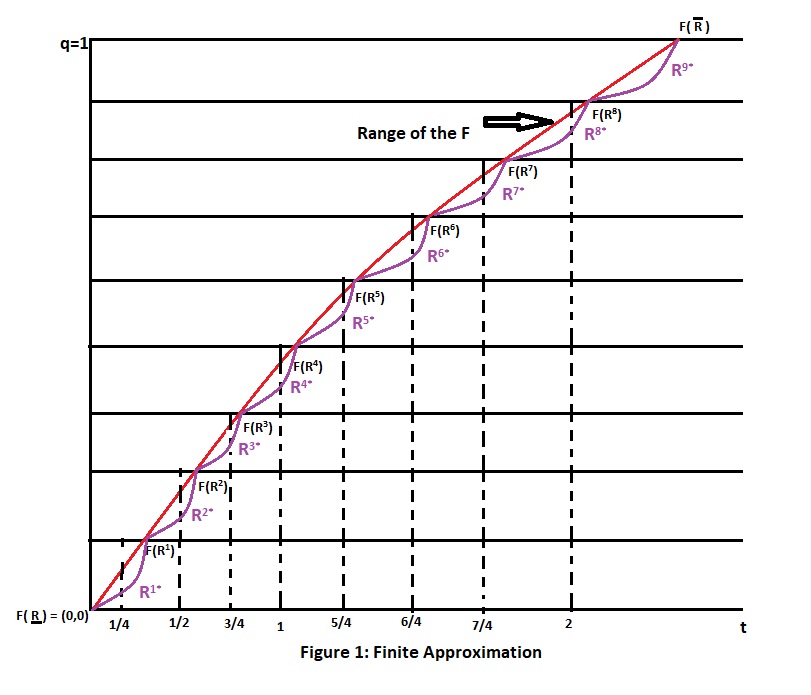}
\end{center}

\noindent{\bf Step $3$:} Now let $t$ be not continuous. Consider $s_2$ again and consider $0\leq t_{F}(R)<\frac{1}{4}$. Let $R^1=\sup\{R\mid t_F(R)<\frac{1}{4}\}$. Suppose, $t_F(R^1)\geq \frac{1}{4}$, then we are not missing $R^1$ while describing the measurable partition of the range of $t_F$. Now, suppose  $t_F(R^1)<\frac{1}{4}$. Then, we have $R^1$ missing while describing the measurable partition of the range of $t_F$ as above. But this does not create any problem with our argument. 
While defining $s_2$ we set $s_2(R)=0$ if $R\in [\underline{R}, R^1]$ and for $\frac{1}{4}$ we have keep the interval open at $R^1$. 
That is, $s_2(R)=\frac{1}{4}$ if $R\in ]R^1,R^2[$. Now, what if there is no $R$ such that $\frac{1}{4}\leq t_{F}(R)<\frac{1}{2}$. This means that the number of elements in the partition of $[\underline{R}, \overline{R} ]$ is less relative to the situation when $t$ is continuous. Thus, let us assume that there is an $R$ such that $\frac{1}{4}\leq t_{F}(R)<\frac{1}{2}$.     
Again depending on what happens at $\frac{1}{2}$ we may also have
$s_2(R)=\frac{1}{4}$ if $R\in ]R^1,R^2]$. That is, if $t$ is not continuous, then intervals could be any of the four kinds described in Section \ref{sec:pre}, otherwise the simple functions that approximate $t_F$ remain step functions as in Step $2$ and the number of elements in the partition maybe less relative to the situation when $t$ is continuous. 

 Now $(0,0)I^{1*}F(R^1)I^{2*}F(R^2)$, and $\underline{R}\precsim R^{1*}\precsim R^1\precsim R^{2*}\precsim R^2$ and so on are well defined by the single-crossing property and monotonicity of $F$. That is, the specifications in $G^{s_2}$ does not depend on the nature of the intervals that define the simple functions. What we need is that for all $R$, $s_2(R)\leq t_{G}^{s_2}(R)$, and this holds. 
 For example consider $[\underline{R} ,R^2]$.  
 For all $R\in [\underline{R}, R^{1*}[$, $s_{2}(R)=0=t_{G}^{s_2}(R)$. Even if $R^{1*}=R^1$ the equality between the two functions holds. 
 Now consider $[R^{1*},R^1]$. 
For all $R\in [R^{1*},R^1], s_{2}(R)\in \{0,\frac{1}{4}\}\leq t_{G}^{s_2}(R)=t_{F}(R^1)$. Let the $F(R^1)<\frac{1}{4}$. 
Then, for all $R\in [\underline{R},R^1]$, $s_{2}(R)=0\leq t_{G}^{s_2}(R)$. Thus, consider $]R^1, R^{2*}[$. For $R\in ]R^1, R^{2*}[$, $s_{2}(R)=\frac{1}{4}\leq t_{G}^{s_2}(R)=t_{F}(R^1)$. For all $R\in [R^{2*}, R^{2}], s_{2}(R)\in \{\frac{1}{4},\frac{1}{2}\}\leq t_{G}^{s_2}(R)=t_{F}(R^2)$ and so on.

\end{proof}
  
 \noindent The following corollary is immediate.
 
\begin{cor}\rm Let $[\underline{R}, \overline{R}]$ be given. Let $F$ be a strategy-proof and individually rational mechanism such that $Rn(F)$ is finite, and for all $G$ with $Rn(G)$ finite $E[F]\geq E[G]$. Then, $F$ is an optimal mechanism among all strategy-proof and individually rational mechanism among all mechanisms.       
 	
\label{cor:1} 	
\end{cor}

\begin{proof} Let by way of contradiction there is a strategy-proof and individually rational mechanism $K$ with $Rn(K)$ not finite and $E[K]>E[F]$. Then, by Theorem \ref{thm:result} it there is $G$, with $Rn(G)$ finite such that $E[F]<E[G]\in ]E[K]-\epsilon, E[K]+\epsilon[$. This is a contradiction.

\end{proof}

 \noindent Corollary \ref{cor:1} implies that if there is an optimal mechanism among the strategy-proof and individually rational mechanisms with finite range, then it must be the optimal mechanism among all mechanisms.       
  
\section{Concluding Remarks}\label{sec:con} In this paper we explore strategy-proof mechanisms defined on closed intervals of rich single-crossing domains. We find that the revenue from any strategy-proof and individually rational mechanism can be approximated by revenues from a sequence of strategy-proof and individually rational mechanisms with finite range. Our proof of the result pertaining to this observation follows from the simple geometry of strategy-proof mechanisms, and the latter follows due to the single-crossing property. Our result applies to both quasilinear and non-quasilinear domains.


\begin{thebibliography}{99}
	\newcommand{\enquote}[1]{``#1''}
	\expandafter\ifx\csname natexlab\endcsname\relax\def\natexlab#1{#1}\fi
	
	
\bibitem[\protect\citeauthoryear{Ash}{2000}]{Ash}
\textsc{Ash. R.,}{2000}~\emph{{Probability and Measure Theory}}, Harcourt Science and Technology Company.
	



\bibitem[\protect\citeauthoryear{Goswami}{2024}]{Goswami1}
\textsc{Goswami. M. P.,}{2024}: \enquote{{Strategy-proof Selling: a Geometric Approach},}\emph{{arXiv:2406.12279}} \url{ 
	https://doi.org/10.48550/arXiv.2406.12279}.



 



\bibitem[\protect\citeauthoryear{Hagerty and Rogerson}{1987}]{Hag and Rog}
\textsc{Hagerty. K. M., and  Rogerson. W. P.,}{1987}\emph{{Robust Trading 
		Mechanisms}}, 
Journal of Economic Theory 42, 94-107  

\bibitem[\protect\citeauthoryear{Kazumura et al.}{2020a}]{Mishra1}
\textsc{Kazumura, T., Mishra, D., and Sherizawa, S. }{2020}: \enquote{{Mechanism Design without Quasilinearlity},}\emph{{Theoretical Economics}}, 15, 511-544.

	
\bibitem[\protect\citeauthoryear{Mirrless}{1971}]{Mirr}
\textsc{Mirrless. J. A.,}{1971}: \enquote{{Optimal Auction Design},}\emph{{Review of Economic Studies}}, 38 175-208. 	



\bibitem[\protect\citeauthoryear{Myerson}{1981}]{myer}
\textsc{Myerson. R.,}{1981}: \enquote{{Optimal Auction Design},}\emph{{Mathematics of Operations Research}}, 6(1), 58-73.


	\bibitem[\protect\citeauthoryear{Rothschild and Stiglitz }{1976}]{Roths}
\textsc{Rothschild. M., and Stiglitz J.R., }{1976}:\enquote{{Equilibrium in competitive insurance markets: an essay in the economics of perfect information},}\emph{{Quarterly Journal of Economics}}, 80, 629-649.


	\bibitem[\protect\citeauthoryear{Spence}{1973}]{Spence}
\textsc{Spence A.M.}{1973}:\enquote{{Job Market Signaling},}\emph{{Quarterly Journal of Economics}}, 
87, 355-374



\end{thebibliography}
\end{document}